\documentclass[12pt]{article}

\usepackage[cmex10]{amsmath}
\usepackage{amsfonts}
\usepackage{amssymb}
\usepackage{amsthm}
\usepackage{graphicx}
\usepackage{multirow}
\usepackage{subfig}
\usepackage{ulem}

\newtheorem{lem}{Lemma}

\newtheorem{theor}[lem]{Theorem}

\begin{document}
%
\title{Expander-like Codes based on Finite Projective Geometry}
%
%
%

\author{\normalsize{Swadesh Choudhary$^1$,
Hrishikesh Sharma$^1$, B.S. Adiga$^2$ and Sachin Patkar$^1$}
\\
\normalsize{$^1$Electrical
        Engineering, IIT Bombay.
         Email: \{patkar,swadesh,hsharma\}@ee.iitb.ac.in} \\
\normalsize{$^2$TCS Innovation Labs, Bangalore.
Email: bs.adiga@tcs.com}
\thanks{Swadesh Choudhary is currently with Stanford University}
\vspace{-.3in}
}
\date{}
\maketitle

\begin{abstract}
We present a novel error correcting code and decoding algorithm which have
construction similar to expander codes. The code is based on a bipartite
graph derived from the subsumption relations of finite projective
geometry, and Reed-Solomon codes as component codes. We use a
modified version of well-known Zemor's decoding algorithm for expander
codes, for decoding our codes.  By derivation of geometric bounds rather
than eigenvalue bounds, it has been proved that for practical values of the
code rate, the random error correction capability of our codes is much better
than those derived for previously studied graph codes, including Zemor's
bound. MATLAB simulations further reveal that the average case
performance of this code is 10 times better than these geometric bounds
obtained, in almost 99\% of the test cases. By exploiting the symmetry of
projective space lattices, we have designed a corresponding decoder that
has \textit{optimal throughput}. The decoder design has been prototyped on
Xilinx Virtex 5 FPGA. The codes are designed for potential applications in
secondary storage media. As an application, we also discuss usage of these
codes to improve the burst error correction capability of CD-ROM decoder.
\end{abstract}
\begin{keywords}
Expander Codes, Bipartite Graph, Finite Projective Geometry.
\end{keywords}


\section{Introduction}
\label{intro_sec}

Graph codes have been studied and analyzed in past, in
order to try and find codes that give good error correction capability at
high code rates \cite{hoholdt}, \cite{janwa}, \cite{zemor}. At the same
time, from a practical implementation point of view, such codes require to
have decoding and encoding algorithms that are efficient in terms of speed and
hardware complexity.

\textit{Expander Codes}, first suggested by Sipser and Spielman \cite{sipser},
proved to be theoretically capable of providing asymptotically good codes that
were decodable in logarithmic time and could be implemented with a circuit
whose size grew linearly with code size. These codes are constructed using
a special graph, known as \textit{expander graph}, and by embedding
identical component codes at the nodes of this graph. \cite{sipser}
analyzed the properties of expander codes, and specified lower bound on
the number of errors that will always be corrected by one decoding
algorithm. Zemor, in his analysis in \cite{zemor}, suggested using
\textit{bipartite Ramanujan} graphs for constructing expander codes. He
also provided a decoding algorithm which improved the lower bound in
\cite{sipser} by a factor of 12. In \cite{hoholdt},
Hoholdt and Justesen built on the work of Tanner on graph codes, by
suggesting the use of Reed Solomon codes as sub-codes for graphs
derived from point-line incidence relations of projective planes. The
decoding speed and ease of implementation, combined with error-correction
performance that was scalable with increasing graph size made all these
codes interesting, while considering applications related to secondary
storage.

By definition, the sub-code length should remain
constant, as the order of (bipartite) expander graph increases, during
construction of a family of expander codes. For
performance reasons described later, we choose to increase the sub-code
length as well. Hence our codes may not be called expander codes, but just graph-based, or \uline{expander-like}
codes. The presented work thus deals with the construction and analysis of an
\textit{expander-like} code, which is based on a special \textit{bipartite
Ramanujan} graph. This bipartite graph is derived from point-hyperplane
incidence relations of projective spaces of \textbf{higher dimensions} than
those suggested by \cite{hoholdt}. We look at various properties of the
codes, and in the process come across several \textit{generic} interesting
properties of projective geometry. Also, we wanted the codes to be
\textit{practically useful}. Hence, in a companion paper, we present 
throughput-optimal VLSI design of decoder for such codes\cite{dmaa_pap}.

For decoding, we employ a variation of Zemor's algorithm. By
simulations using this algorithm, we
found that the codes have excellent robustness to \textit{random} as well as
\textit{burst} errors. Hence we envisage their application in data
storage systems. 

The next section provides the basic properties of cardinalities of
projective geometry that we will be using. Section 3 gives relevant
background information for the various concepts required in this paper.
Section 4 describes our code construction in detail. Sections 5-11 give the
characterization of error-correction capability for these codes, including
proofs of propositions relating to bounds on error correction capacity. The
remaining sections detail out the prototyping results, and also the
application to two types of storage media (namely CD-ROMs and DVD-R),
before concluding the paper.
 
\section{Finite Projective Spaces}
\label{pg_intro}
In this section, we look at how finite projective spaces are generated from
finite fields. For more details, refer \cite{hrishi}. Consider a finite field $\mathbb{F} =$ $\mathbb{GF}(s)$
with $s$ elements, where $s$ is a power of a prime number $p$ i.e. $s=p^{k}$,
$k$ being a positive integer. A projective space of dimension $d$ is denoted
by ${\mathbb{P}}(d,\mathbb{F})$ and consists of one-dimensional subspaces
of the $(d+1)$-dimensional vector space over $\mathbb{F}$ (an extension
field over $\mathbb{F}$), denoted by $\mathbb{F}^{d+1}$. Elements of this
vector space are of the form $(x_{1},\ldots,x_{d+1})$, where each $x_{i}
\in \mathbb{F}$. The total number of such elements are $s^{(d+1)}$ =
$p^{k(d+1)}$. An equivalence relation between these elements is defined as
follows. Two non-zero elements ${\bf{x}}$, ${\bf{y}}$ are
\textit{equivalent} if there exists an element $\lambda \in$ GF$(s)$ such
that ${\bf{x}}=\lambda {\bf{y}}$. Clearly, each equivalence class consists
of $s$ elements of the field ($s-1$ non-zero elements and ${\bf{0}}$), and
forms a one-dimensional subspace. Such 1-d vector subspace corresponds to
a \textbf{point} in the projective space. Points are the zero-dimensional
subspaces of the projective space. Therefore, the total number of points
in ${\mathbb{P}}(d,\mathbb{F})$ are
\begin{eqnarray}
P(d) &=& \frac{\textrm{\# non-zero elements in the field}}
{\textrm{\# non-zero elements in one equivalence class}} \\
&=& \frac{s^{d+1}-1}{s-1}
\end{eqnarray}
An $m$-dimensional subspace of ${\mathbb{P}}(d,\mathbb{F})$ consists
of all the one-dimensional subspaces of an $(m+1)$-dimensional subspace
of the vector space. The basis of this vector subspace will have $(m+1)$
linearly independent elements, say $b_{0},\ldots,b_{m}$. Every element of
this subspace can be represented as a linear combination of these basis
vectors.
\begin{equation}
{\bf{x}} = \sum_{i=0}^{m} \alpha_{i} b_{i}, \textrm{ where } \alpha_{i} \in \mathbb{F}(s)
\end{equation}
Clearly, the number of elements in the vector subspace are $s^{(m+1)}$.
The number of points in the $m$-dimensional projective subspace
is given by $P(m)$ defined in earlier equation.
Various properties such as degree etc. of a {\normalsize
$\mathbf{m}$}-dimensional projective subspace \textbf{remain same}, when
this subspace is \textit{bijectively} mapped to some {\normalsize $(\mathbf{d-m-1})$}-dimensional
projective subspace. The two sets of these subspaces,
one for each dimension, are said to be \uline{\textit{dual of each other}}.
The number of \textbf{d}-dimensional projective subspaces of a
\textbf{m}-dimensional projective space can be counted using the
\textit{Gaussian Coefficient}.
\begin{equation}
\phi(n,l,s)=\frac{(s^{n+1}-1)(s^{n}-1)\ldots(s^{n-l+1}-1)}{(s-1)(s^{2}-1)\ldots(s^{l+1}-1)}
\end{equation}
For {\normalsize $0 \leq l<m \leq d$}, the number of {\normalsize
$\mathbf{l}$}-dimensional projective subspaces contained in an {\normalsize
$\mathbf{m}$}-dimensional projective subspace is {\normalsize $\phi(m,l,s)$}, while the number of {\normalsize
$\mathbf{m}$}-dimensional projective subspaces containing a particular
{\normalsize $\mathbf{l}$}-dimensional projective subspace is {\normalsize
$\phi(d-l-1,m-l-1,s)$}.

\subsection{Projective Spaces as Lattices}
It is a well-known fact that the lattice of subspaces in any projective
space is a \textbf{modular, geometric lattice} \cite{hrishi}. A
projective space of dimension 2 is shown in figure \ref{pg_lat}.
In the figure, the top-most node represents the \textit{supremum},
which is a projective space of dimension {\normalsize \textbf{m}} in a lattice
for {\normalsize $\mathbb{P}(\mathbf{m},\mathbb{GF}(\mathbf{q}))$}. The
bottom-most node represents the \textit{infimum}, which is a projective
space of (notational) dimension -1. Each node in the lattice as such is a
projective subspace, called a \textbf{flat}. Each horizontal level of flats
represents a collection of all projective subspaces of {\normalsize
$\mathbb{P}(\mathbf{m},\mathbb{GF}(\mathbf{q}))$} of a particular
dimension. For example, the first level of flats above infimum are flats of
dimension 0, the next level are flats of dimension 1, and so on. Some
levels have special names. The flats of dimension 0 are called
\textit{points}, flats of dimension 1 are called \textit{lines}, flats of
dimension 2 are called planes, and flats of dimension (m-1) in an overall
projective space {\normalsize $\mathbb{P}(\mathbf{m},\mathbb{GF}(\mathbf{q}))$}
are called \textit{hyperplanes}.
\begin{figure}[h]
\begin{center}
\subfloat[PG graph]{\label{fold_bp}\includegraphics[scale=0.3]{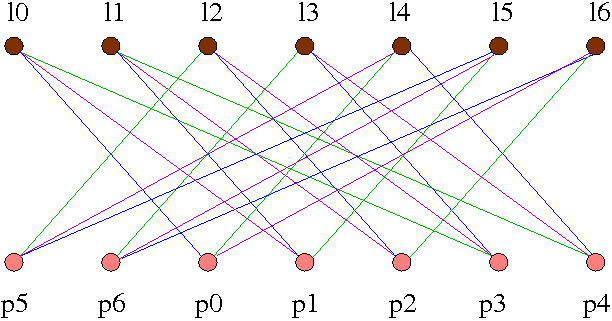}}
\quad
\subfloat[Corresponding Lattice
Diagram]{\label{pg_lat}\includegraphics[scale=0.5]{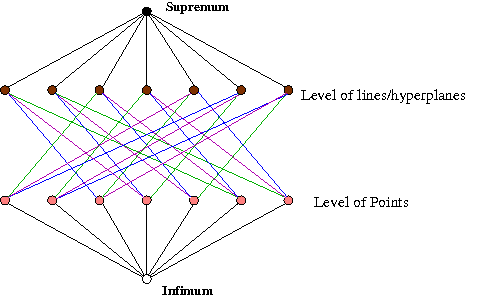}}
\end{center}
\caption{An Example PG Bipartite Graph}
\label{pg_gph}
\end{figure}

\section{Expander Codes}
\label{exp_code}
Expander codes are a family of asymptotically good, linear error-correcting
codes \cite{sipser}. They can be decoded in sub-linear time(proportional to
\textbf{log(n)}, where \textbf{n} is length of codeword) using
\textit{parallel} decoding algorithms. Further, this can be achieved using
identical \textit{component} decoders, whose count is 
proportional to \textbf{n}.  These codes are based on a class of graphs
known as \textit{expander graphs}. One construction of
expander graph, used in construction of expander codes, is by considering
the edge-vertex incidence graph $\mathbb{Z}$ of a \textbf{d}-regular graph $\mathbb{G}$.
The edge-vertex incidence graph of $\mathbb{G = (V,E)}$, a (2,d)-regular
bipartite graph, has vertex set $E\cup V$ and edge set 
\[ \left\{(e,v) \in E \times V:\; \mbox{v is an endpoint of e}\right\}\]
Vertices of $\mathbb{Z}$ corresponding to edges E of $\mathbb{G}$ are then
associated to \textit{variables}, while vertices of $\mathbb{Z}$
corresponding to vertices of $\mathbb{G}$ are associated to constraints on
these variables. Each constraint corresponds to a set of \textit{linear}
restrictions on the \textbf{d} variables that are its neighbors. In particular, a constraint will require that the
variables it restricts form a codeword in some linear code of
length \textbf{d}. This constraint forms the \textbf{sub-code} for that expander
code. Further, all the constraints are required to impose
\textbf{isomorphic} codes (on different variables, of course). The default
construction of a family of expander codes requires \textbf{d} to remain
constant, as the order of $\mathbb{G}$ increases.

Formally, Let $\mathbb{Z}$ be a (2,d)-regular graph between set of
\textbf{n} nodes called \textit{variables}, and $\frac{2}{d}$\textbf{n}
nodes called \textit{constraints}. Let \textit{z(i,j)} be a function such
that for each constraint $C_i$, the variables neighboring $C_i$ are
$v_{z(i,1)},\cdots,v_{z(i,d)}$. Let $\mathbb{S}$ be an error-correcting
code of block length \textbf{d}. The expander code $\mathbb{C(Z,S)}$ is the
code of block length \textbf{n} whose codewords are the words
$(x_1,\cdots,x_n)$ such that, for 1 $\leq$ $i$ $\leq$
$\frac{2}{d}$\textbf{n}, $x_{z(i,1)},\cdots,x_{z(i,d)}$ is a codeword of
$\mathbb{S}$.

\subsection{Expander Graphs}
\label{exp_graph}
An expander graph is a graph in which every set of vertices has an
unusually large number of neighbors. More formally, \\
Let $\mathbb{G}$ = ($\mathbb{V,E}$) be a graph with \textbf{n} vertices.
Then the graph $\mathbb{G}$ is a $\delta$-expander, if every set of
\textit{at most} \textbf{m} vertices expands by a factor of $\delta$. That
is,
\[
\forall S \subset V: \left|S\right| \leq m  \Rightarrow 
\left| \left\{ y : \exists x \in S\mbox{ s.t. } (x,y) \in E
\right\}\right| > \delta\cdot\left|S\right|
\]
Expander codes being a subclass of LDPC codes, for whom efficient iterative
decoding using variables and constraints a bipartite graph is feasible, we are
interested mainly in bipartite expander graphs.

The degree of ``goodness" of expansion, especially for regular graphs, can
also be measured using its eigenvalues. The largest eigenvalue of a
\textbf{k}-regular graph is \textbf{k}. If the second largest eigenvalue is
\textit{much smaller} from the first (\textbf{k}), then the graph is known
to be a good expander \cite{sipser}.

\subsection{Good Expander Codes}
\label{good_code_sec}
As pointed out earlier, the decoding algorithm for such codes is
\textit{iterative}. Hence good expander codes imply \textit{at least} the
following properties.
\begin{itemize}
\item Better minimum distance than other codes of same length,
\item Fast convergence, and
\item Better code rate than other codes in the same class
\end{itemize}

Good codes having the above properties can be identified with help of three
theorems proved in \cite{sipser}. For the theorems, we assume that an
expander code $\mathbb{C(Z, S)}$ has been
constructed having $\mathbb{S}$ as a linear code of rate \textbf{r}, block
length \textbf{d}, and minimum distance $\epsilon$, and
$\mathbb{Z}$ as the edge-vertex incidence graph of a \textbf{d}-regular
graph $\mathbb{G}$ with second-largest eigenvalue $\lambda$.

\begin{theor}
\label{1theor}
The code $\mathbb{C(Z, S)}$ has rate at least
2\textbf{r} - 1, and minimum relative distance at least
$\left(\frac{\epsilon - \frac{\lambda}{d}}{1 - \frac{\lambda}{d}}\right)^{2}$.
\end{theor}

\begin{theor}
\label{2theor}
If a parallel decoding round for $\mathbb{C(Z,S)}$, as given in
\cite{sipser}, is given as input a word of relative distance $\alpha$ from a
codeword, then it will output a word of relative distance at most
$\alpha\left(\frac{2}{3} + \frac{16\alpha}{\epsilon^2} +
\frac{4\lambda}{\epsilon d} \right)$ from that codeword.
\end{theor}

\begin{theor}
\label{3theor}
For all $\epsilon$ such that 1- 2H($\epsilon$) $>$ 0, where H($\cdot$) is
the binary entropy function, there exists a polynomial-time
constructible family of expander codes of rate 1 - 2H($\epsilon$) and
minimum relative distance arbitrarily close to $\epsilon^2$ in which any
$\alpha$ $<$ $\epsilon^2$/48 fraction of error can be corrected by a circuit
of size O(n log n) and depth O(log n).
\end{theor}

From theorem \ref{1theor}, we observe that to have high minimum relative
distance, we should have \textit{$\epsilon$ as high}, and \textit{$\frac{\lambda}{d}$ as
low}. Since $\mathbb{Z}$ has been constructed out of \textbf{d}-regular
graph $\mathbb{G}$, low $\frac{\lambda}{d}$ signifies high distance between
first and second eigenvalues, i.e. the graph $\mathbb{G}$ has to be a
``good" expander graph. Further, to have high rate, $\mathbb{S}$ has to
have a \textit{high rate \textbf{r} as well}, other than having \textit{high minimum
relative distance $\epsilon$}.

From theorem \ref{2theor}, we observe that to shrink the distance of input
word after one iteration maximally, we need to again have \textit{$\epsilon$ as
high} as possible, and \textit{$\frac{\lambda}{d}$ as low} as possible. Such maximal
shrinking of distance, per iteration, leads to the fastest convergence
possible, as is also brought out in the proof of theorem \ref{3theor}.

From theorem \ref{3theor}, we observe that to be able to correct as high
fraction of errors as possible, we need to have \textit{$\epsilon$ as high} as
possible, again.

As an aside, $\frac{\lambda}{d}$ is low whenever \textbf{d} is high. In
PG-based bipartite graphs, \textbf{d} does increase, as \textbf{n}
increases (hence \textit{expander-like} codes), which hence helps in
making the overall code advantageous over classical expander codes.

\subsection{Good Expander Graphs}
Zemor pointed out \cite{zemor} that if $\mathbb{G}$ is a bipartite
graph, then the \% of random errors that can be corrected using a parallel
iterative decoding algorithm can be \textbf{increased twelve-fold}. He also
pointed out that the upper bound on minimum distance, as pointed
out in theorem \ref{1theor}, can also be achieved faster, if one considers
\textit{Ramanujan graphs} (since $\frac{\lambda}{d}$ value is low for these
graphs). Overall, he suggested using bipartite Ramanujan graphs for
construction of good expander codes instead.

\section{Details of Expander-like Code}
\label{our_code}

\subsection{PG Graphs as Good Expander Graphs}
Balanced regular bipartite graphs $G_{d,d}$, which are \textit{symmetric
balanced incomplete block designs(BIBDs)} are known to be Ramanujan
graphs \cite{janwa}. Incidence relations of projective geometry
structures give such BIBDs, and hence Ramanujan
graphs. Usage of projective plane as $G_{d,d}$ along with RS
codes as component codes to construct good expander-like graph codes was
first reported in \cite{hoholdt}. For our work, we \textbf{do not limit}
to projective planes: to have better performance, we have made use of
point-hyperplane incidence graphs from higher dimensional projective
spaces, which also satisfy the eigenvalue properties that make a Ramanujan
graph.  Some reasons for using projective geometry are as follows.
\begin{itemize}
\item As detailed in a companion paper \cite{dmaa_pap}, the mapping of
vertices to points and hyperplanes enables us to use several projective
geometry properties for \textit{disproving the existence of certain bipartite
subgraphs of a fixed minimum degree}. This strategy leads us to finding the
minimum number of vertices required to form a complete bipartite subgraph
of a given minimum degree. This number of vertices is required to
calculate tight \textbf{geometric bounds} for error correction capability
of the overall code. Thus, we don't need to use complicated eigenvalue
arguments used by \cite{sipser} and \cite{zemor}. Also, the bounds obtained
in this manner are better than our predecessors. Furthermore, Zemor had
restricted the subcodes to be constrained by $d\geq3\lambda$, $\lambda$
being the second largest eigenvalue of the graph. We have \textit{no such
restriction}.
\item The use of projective geometry also helps in developing a perfect
\textit{folded architecture} of the decoder for hardware implementation,
discussed later in section \ref{results_sec}. This particular folding
enables efficient utilization of processors and memories, while being
throughput-optimal.
\end{itemize}

\subsection{Reed-Solomon Codes as Good Component Codes}
By choosing a ``good expander" graph, and fixing a code with high
minimum relative distance $\epsilon$, one can design code having the first
two properties described in section \ref{good_code_sec}. To simultaneously
have high code rate for $\mathbb{C(Z,S)}$, the component code $\mathbb{S}$
also needs to have high rate \textbf{r}. 
\textit{Reed-Solomon} codes are a class of non-binary,
linear codes, which for a given rate, have the best minimum relative
distance(so-called \textit{maximum distance separable} codes), \textbf{and
vice-versa}. Hence we use RS codes as the sub-codes to our expander-like
codes.
\subsection{Code Construction}
\label{code_sec}
To construct an expander-like code, we follow \cite{zemor}. We generate a
\textbf{balanced regular bipartite graph} $\mathbb{G}$ from a projective space. A
projective space of dimension \textbf{n} over $\mathbb{GF}(2)$,
$\mathbb{P}(n,\mathbb{GF}(q))$, has at least following two properties,
arising out of inherent duality:
\begin{enumerate}
\item The number of subspaces of dimension $m$ is equal to the number of
subspaces of dimension ($n-m-1$).
\item The number of $m$-dimensional subspaces incident on each
($n-m-1$)-dimensional subspace is equal to the number of ($n-m-1$)-dimensional
subspaces incident on each $m$-dimensional subspace.
\end{enumerate}
We associate one vertex of the graph with each
$m$-dimensional subspace and one with each ($n-m-1$)-dimensional subspace.
Two vertices are connected by an edge if the corresponding subspaces are
incident on each other. As edges lie only between subspaces of
different dimensions, the graph is bipartite with vertices associated with
$m$-dimensional subspaces forming one set and vertices associated with
($n-m-1$)-dimensional subspaces forming another. Also, the two properties,
listed above, ensure that both the vertex sets have the same number of
elements and that each vertex has the same degree. To be able to quantify
various properties of the constructed code, we hereafter specifically
consider the graph, $\mathbb{G}=(V,E)$, obtained by taking the points and
hyperplanes of $\mathbb{P}(5,\mathbb{GF}(2))$. In this projective
space, the number of points (= number of hyperplanes) is $\phi(5,0,2)=63$.
Each point is incident on $\phi(4,3,2)=31$ hyperplanes and each hyperplane
has $\phi(4,0,2)=31$ points. Therefore, we have $|V|=126$ and $|E|=1953$.
This implies that the block length of code $\mathbb{C}$ is $1953$ and the
number of constraints in the code is $126$. The second eigenvalue of
$\mathbb{G}$, $\lambda$ is 4, according to a formulae by \cite{chee}. Hence
the ratio $\frac{\lambda}{d}$ is quite small, as required for design of
``good" expander codes.

As the expander graph $\mathbb{G}$ is $31$-regular, the block length of the
component code must also be $31$ \cite{sipser}. We choose the $31$-symbol
\textit{shortened Reed Solomon code} as the component code, with each
symbol consisting of eight bits. To
have performance advantage, we also \uline{modify Zemor's decoding
algorithm} \cite{zemor} as follows. If a particular vertex detects more
errors than it can correct, we skip the decoding for that vertex. In
Zemor's algorithm, the decoding is still carried out in such case, which
can lead to possibly getting a (different) codeword with more errors.
The modification is possible because it is possible to compute, as a
\textit{side output} using Berlekamp-Massey's algorithm for RS
decoding\cite{inf_th}, whether the degree of errors in the current input
block of symbols to the decoder be corrected or not. If not, then the
algorithm can be made to skip decoding, thus preserving the errors in the
input block. This \textit{variation in decoding} will reduce the number of
extra errors introduced by that vertex if the decoding fails.  Based on
this decoding algorithm, a MATLAB model of decoder was first made, to
observe code's performance as discussed next. 

\section{Performance of Code for Random Errors}
\label{rand_perf}

To benchmark the error-correction performance in wake on \textit{random
errors}, we varied the minimum distance of
the component code, and simulated the MATLAB decoder model. Random symbol
errors were introduced at random locations of the zero vector. Convergence
of the decoder's output back to the zero vector was checked after
simulation. As our code is linear, the performance obtained in testing for
zero vector is valid for the entire code. Since the errors were introduced
at random locations, simulations were run over \textit{many different
rounds} of decoding for different \textit{pseudo-random sequences} as
inputs, and \textit{averaged}, to get reliable results. These sequences differ
in random positions in which the errors are introduced. Each round of decoding
for particular input further involves several iterations of execution of
decoding algorithm. One iteration of decoding corresponds to both sides
of the bipartite graph to finish decoding the component codes.

It is \textit{observed experimentally} that in case of a decoding failure,
beyond 4 iterations, the number of errors in the codeword
stabilizes(referred to as \textit{fixed point} in \cite{barg-zemor}). In
the first few iterations, as the number of errors decrease in the overall
code in a particular iteration, the number of errors \textit{on average} to
be handled by RS decoders in next iteration is lesser. Hence
probabilistically, and experimentally, these component decoders converge
faster, thus \textit{increasing} the percentage of errors being corrected
in its next iteration. However, after maximum 4 iterations, it was seen
that there is no further reduction in errors. This phenomenon could most
probably be attributed to infinite oscillations of errors in an embedded
subgraph, to be described in next section.

Hence we have fixed the stopping threshold of decoder to \textit{exactly 4
iterations} not only in simulation, but also in the practical design of a
CD-ROM decoder. In simulation, if there are non-zero entries remaining
after 4 iterations, then the decoding is considered to have failed. In real
life, if one or more component RS decoders fail to converge at 4th
iteration, then again decoding is considered to have failed. The results of
our simulations are presented in Tables \ref{tab:1} and \ref{tab:2}. The
component codes used for these simulations have minimum distance as 5 and 7,
respectively. The ``failures" column represents the percentage of failed
decoding attempts. The ``average number of iterations" column signifies
the \textit{average} number of iterations required for successful decoding
of a corrupt codeword, over various rounds.

\begin{table}[h]
\centering
\begin{tabular}{|c|c|c|}
\hline
No. of errors & failures & Avg. no. of iterations\\ \hline
50 & 0 & 1 \\ \hline
80 & 1 & 1.71 \\ \hline
100 & 18 & 2.33 \\ \hline
110 & 40 & 2.72  \\ \hline
\end{tabular}
\caption{Random errors ($\epsilon=5$)}
\label{tab:1}
\end{table}

\begin{table}[h]
\centering
\begin{tabular}{|c|c|c|}
\hline
No. of errors & failures & Avg. no. of iterations\\ \hline
150 & 0 & 1.6   \\ \hline
175 & 0 & 1.99  \\ \hline
200 & 0 & 2.19  \\ \hline
250 & 23 & 3.82 \\ \hline
275 & 64 & 4.5  \\ \hline
\end{tabular}
\caption{Random errors ($\epsilon=7$)}
\label{tab:2}
\end{table}
We present some worst-case bounds on rate and error-correction capability
of our codes in Table \ref{tab:3}. We vary the minimum distance of subcode
between 3 and 15. Beyond 15, rate of the overall code $\mathbb{C}$ becomes
very less and hence impractical. We also compare these bounds to the
bounds derived by Zemor. For calculating the Zemor bounds and making a
fair comparison, we need to remove the advantage of using Reed Solomon
codes as sub-codes. Zemor had derived the bounds for general codes
assuming that $\geq$ $\frac{\epsilon}{2}$ errors could not be
corrected for \textbf{any} distance(even/odd). For Reed Solomon component
codes used in our construction, since we use only \textit{odd} distances, 
$\frac{(\epsilon+1)}{2}$ errors can never be corrected. To
account for this, we replace $\epsilon/2$ by $(\epsilon+1)/2$ in Zemor's
formula to calculate the bounds.
\begin{table}[h]
\centering
\begin{tabular}{|c|c|c|c|c|}
\hline
Min. dist. & Subcode &  Lower bound &  Error-correcting  &
Zemor's \\
subcode $(\epsilon)$ & rate & on rate of $\mathbb{C}$ & capability & bound \\
$(\epsilon)$ & & & of $\mathbb{C}$ & for $\mathbb{C}$ \\  \hline
3 & 0.94 & 0.87 & 3 & -- \\\hline
5 & 0.87 & 0.74 & 8 & -- \\ \hline
7 & 0.81 & 0.61 & 15 & -- \\ \hline
9 & 0.74 & 0.48 & 24 & -- \\ \hline
11 & 0.68 & 0.35 & 35 & -- \\ \hline
13 & 0.61 & 0.23 & 48 & 42 \\ \hline
15 & 0.55 & 0.1 & 87 & 65 \\ \hline
\end{tabular}
\caption{Change in parameters of $\mathbb{C}$ with variation in
         minimum distance of subcode}
\label{tab:3}
\end{table}

We give a geometrical analysis of process of error correction in the
overall code $\mathbb{C}$. We have also used results from this analysis 
to derive the bounds on error correction capability
of $\mathbb{C}$. First of all, since we have $\mathbb{P}(5,\mathbb{GF}(2))$,
points form the 0-dimensional subspace and hyperplanes form the
4-dimensional subspace. Moreover, planes form 2-dimensional subspaces of the
projective space, and are symmetric with respect to points and hyperplanes.
Finally, 7 points are contained in a plane and a plane is contained in 7
hyperplanes in $\mathbb{P}(5,\mathbb{GF}(2))$.

To understand the limits, given the minimum distance $\epsilon$ of the
subcode, we need to find the minimum number of random errors to be
introduced in $\mathbb{C}$, which will cause the decoding to fail. This
will happen if the vertices corresponding to the points and hyperplanes
\textbf{get locked} in such a way that in each iteration, an \textbf{equal
number} of constraints fail on each side. This is the
\textit{\textbf{minimal} configuration of failure}. As explained next,
errors can expand over iterations(more edges represent corrupt symbols),
but that is not the minimum configuration of failure. Similarly, if errors
shrink, i.e. lesser number of vertices in bipartite graph fail in next
iteration, then it leads to a case of decoding convergence, not decoding
failure.


For example, if we consider $\epsilon=5$, each vertex of the graph can
correct up to 2 erroneous symbols($\lfloor\frac{\epsilon}{2}\rfloor$) in the
set of symbols that it is decoding.  If 3 or more erroneous symbols are
given to it, then either the decoder, based on Berlekamp-Massey's algorithm,
skips decoding, or it outputs another codeword that in worst case has at
least $\epsilon$ different symbols now(than the transmitted codeword), and
hence at least $\epsilon$ errors. However, as discussed earlier, we are not
interested in the latter case. So, if we can generate a case in which
decoding of the sub-code fails at vertices corresponding to 3 points, all
of which are incident on 3 hyperplanes, we have a situation in which the 3
points will transfer at least 3 errors to each of vertices corresponding to
the 3 hyperplanes. These vertices may also fail, or decode a different
codeword, while decoding their inputs. Again in the worst case each of
these hyperplane decoders will output at least 3 erroneous symbols. These
corrupt symbols, or errors, are then transferred back to the vertices
representing the 3 points. Thus, the errors will keep \textbf{oscillating
infinitely} from one side of the graph to the other, and the decoder will
never decode the right codeword. Thus, a minimum of $3*3=9$ errors are
required to cause a failure of decoding. This can be seen from the figure
\ref{9_fail}.
\begin{figure}[!h]
\centering
\includegraphics[scale=0.3]{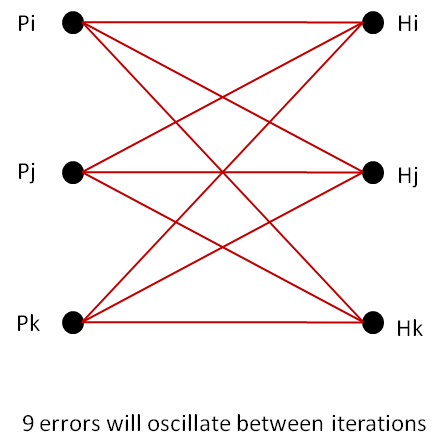}
\caption{Subgraph that causes failure ($\epsilon=5$)}
\label{9_fail}
\end{figure}
For any case in which less than 9 corrupt symbols exist, by pigeonhole
principle, we will have \textit{at least one} hyperplane or point having
less than 3 errors incident on it. Decoder corresponding to that vertex
will correctly decode the sub-code, thus reducing the total number of
errors flowing in the overall decoder system of $\mathbb{C}$. This will, in
next iteration, cause some other hyperplane or point to have less than 3
errors. Thus in the subsequent iterations, all the errors will definitely
be removed. Therefore, \textit{8 errors or less will always be corrected}.
This has been illustrated in figure \ref{8_success}.
\begin{figure}[!h]
\centering
\includegraphics[scale=0.2]{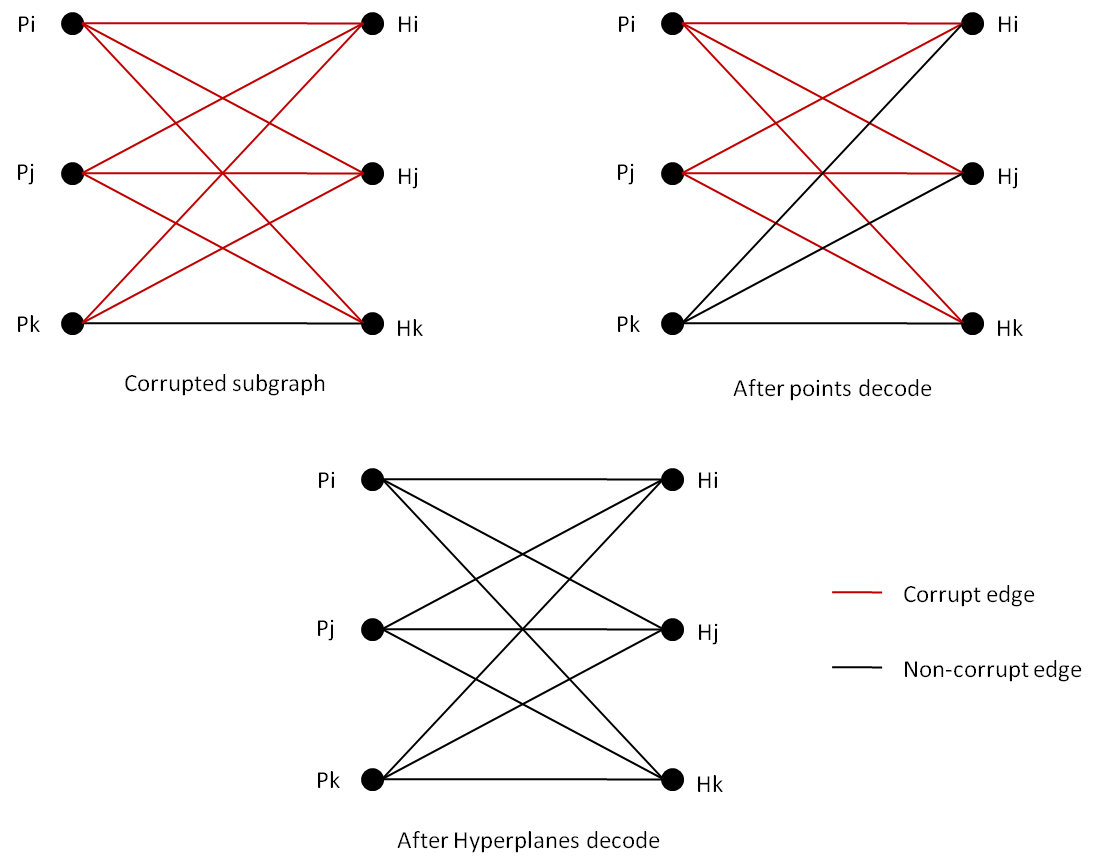}
\caption{8 or less errors always corrected}
\label{8_success}
\end{figure}
As we can see from Table ~\ref{tab:1}, the worst case scenario is very
unlikely to occur and for randomly placed errors, even 80 errors are found
to be corrected 99\% of the time.

Now the question is, when can we find a configuration in which 3 points
are all incident on 3 hyperplanes? If we choose any plane in the given
geometry, we can pick up any three points of that plane and find any 3
hyperplanes corresponding to the same plane. This will ensure that all the
3 points are incident on all the 3 hyperplanes. Thus, if for some input,
the decoding at these 3 points fails, in the worst case they will corrupt
all the edges incident on them. This in turn would cause 3 errors each, on
the chosen hyperplanes. Hence the errors would oscillate between points and
hyperplanes for each successive iteration. Thus, in the worst case, there
need to be 9 erroneous symbols, located such that they are incident on the
3 chosen points, to cause the decoding of $\mathbb{C}$ to fail.

In general, if we are given a minimum distance $\epsilon$ of the subcode,
it is known that at each vertex, more than ($\frac{\epsilon+1}{2}$) errors
will not be corrected. So, in our graph we need to find the \textit{minimum}
number of vertices $\xi$ required to get a \textbf{embedded bipartite
subgraph} such that each vertex in the subgraph has a degree of at least
($\frac{\epsilon+1}{2}$) towards vertices on other side of the subgraph.
Once this number of vertices in some embedded subgraph has been found, the
number of errors that can always be corrected by our decoder is given by:
\begin{equation*}
E = \xi(\frac{\epsilon+1}{2}) - 1
\end{equation*}
In $\mathbb{P}(5,\mathbb{GF}(2))$, a plane has 7 points, and is contained
in 7 hyperplanes. For $3 \leq \epsilon \leq 13$, $\epsilon$ being odd, the
minimum number of vertices $\xi$ corresponds to ($\frac{\epsilon+1}{2}$),
and the corresponding points and hyperplanes can be picked from any plane.
For $\epsilon\geq15$, the calculation of $\xi$ is non-trivial, since points
and hyperplanes from one plane are not sufficient. This is because
$\frac{\epsilon+1}{2}=8$, which would require us to get a subgraph
of minimum degree 8. Construction of such embedded subgraph is  not
possible by choosing only one plane.
In next few sections, we give detailed proof to show non-existence of a
minimum degree 8 embedded bipartite subgraph having 9, 10 vertices on each side,
within point-hyperplane graph of $\mathbb{P}(5,\mathbb{GF}(2))$. Assuming that a sub-graph with 11 vertices
exists, having a vertex degree of at least 8 (we have not been able to
disprove it yet), we get the bound stated in table \ref{tab:3} for
$\epsilon=15$. Since our constructions are exact, we can use these
\textit{tight} lower bounds for all practical values of $\epsilon$,
wherever calculation of it is possible. Otherwise, another \textit{looser},
lower bound can be found using Lemma 1 of \cite{hoholdt}, but derivation of
that uses eigenvalue calculations. Without that, it is still clear from
this table that the bound is \textbf{much better} than the bound obtained
by Zemor\cite{zemor} using eigenvalue arguments.

\subsection{Bound in case of $\epsilon=15$}
\label{props}

Earlier in this section, we highlighted the method to analyze the error
correction bounds for our code. We also derived the
bounds for $3 \leq \epsilon \leq 13$. For the complex case of
$\epsilon=15$, we present in next few sections, overview of proofs for
establishing the non-existence of certain sub-graph embeddings which help
us improve the bound, as reflected in table \ref{tab:3}. The
detailed proofs are presented in Appendix \ref{proof_app}. Appendix \ref{eigen} presents
a eigenvalue based approach (similar to Zemor's arguments).

We first describe the two
theorems, that prove the non-existence of certain minimal embeddings.

\begin{theor}
\label{theor1}
In the construction of bipartite graph mentioned above, there exists no
embedded subgraph having size of partitions as 9, the degree of each of
whose vertices has a minimum degree($\delta$) of 8.
\end{theor}
\begin{theor}
\label{theor2}
In the construction of bipartite graph mentioned above, there exists no
embedded subgraph having size of partitions as 10, the degree of each of
whose vertices has a minimum degree($\delta$) of 8.
\end{theor}




\section{Vector Space Representation of Geometry}
In this section, we outline certain details used in the proofs of
the above theorems, in Appendix \ref{proof_app}.
To recall, the points of an $n$-dimensional projective space over a field
$\mathbb{F}$ can be taken to be the equivalence classes of nonzero vectors
in the (n + 1)-dimensional vector space over $\mathbb{F}$. Vectors in an
equivalence class are all scalar multiples of one-another. These vector
being one-dimensional subspaces, they also represent the rays of a vector
space passing through origin. The orthogonal subspace of each such ray
is the {\bf unique} n-dimensional subspace of $\mathbb{F}^{n + 1}$, known as
hyperplanes. Each vector $h$ of such orthogonal subspace is linked to the
ray, $p$, by a dot product as follows.
\begin{equation*}
    p_{0}h_{0} + p_{1}h_{1} + \cdots + p_{n}h_{n} = 0
\end{equation*}
where $p_i$ is the $i^{th}$ coordinate of $p$. This uniqueness implies
bijection, and hence a vector $p$ can be used to represent a hyperplane
subspace, which is exclusive of this vector as a point. Due to duality,
similar thing can be said about a hyperplane subspace.

Hereafter, whenever we say that two projective subspaces of same dimension
are independent, we mean the linear independence of the corresponding
vector subspaces in the overall vector space.

\subsection{Relationship between Projective Subspaces}
Throughout the remaining paper, we will be trying to relate projective
subspaces of various types. We define the following \textbf{terms} for
relating projective subspaces.
\begin{description}
\item[\textbf{[Contained in]}] ~\\If a projective subspace \textbf{X} is said to be
        contained in another projective subspace \textbf{Y}, then the
        vector subspace corresponding to \textbf{X} is a vector subspace
        itself, of the vector subspace corresponding to \textbf{Y}.
        In terms of projective
        spaces, the points that are part of \textbf{X}, are also part of
        \textbf{Y}. The \uline{inverse relationship} is termed
        `\textbf{contains}', e.g. ``\textbf{Y} contains \textbf{X}''.
\item[\textbf{[Reachable from]}] ~\\If a projective subspace \textbf{X} is said to be
        reachable from another projective subspace \textbf{Y}, then
        \textit{there exists} a chain(path) in the corresponding lattice
        diagram of the projective space, such that both the flats,
        \textbf{X} and \textbf{Y} lie on that particular chain. There is no
        directional sense in this relationship.
\end{description}
\section{Lemmas Used in Proof}
In this section, we describe certain useful lemmas, related to proofs of
theorems \ref{theor1} and \ref{theor2}. These short lemmas and their proofs
provide an insight into the way detailed proofs were built up for theorems
\ref{theor1} and \ref{theor2}, in appendix \ref{proof_app}.

\subsection{Lemmas Related to Projective Space}
\begin{lem}
\label{4ind}
In projective spaces over $\mathbb{GF}$(2), any subset of
points(hyperplanes) having cardinality of 4 or more has 3
non-collinear(independent) points(hyperplanes).
\end{lem}
\begin{proof}
The underlying vector space is constructed over GF(2). Hence, any
2-dimensional subspace of contains the zero vector, and non-zero vectors
of the form $ \alpha a+ \beta b$. Here, $a$ and $b$ are linearly independent
one-dimensional non-zero vectors, and $\alpha$ and $\beta$ can be either 0 or
1:\\
$\mbox{ }$
\(\alpha,\beta\;\in\;\mathbb{GF}(2):(\alpha\;=\;\beta)\;\neq\;0\).\\
Thus, any such 2-d subspace contains exactly 3 non-zero vectors. Therefore, in
any subset of 4 or more points of a projective space over $\mathbb{GF}(2)$
(which represent one-dimensional non-zero vectors in the corresponding vector
space), at least one point is not contained in the 2-d subspace formed by 2
randomly picked points from the subset. Thus in such subset, a further subset
of 3 independent points(hyperplanes) i.e. 3 non-collinear vectors can always
be found.
\end{proof}

\begin{lem}
\label{3overlap}
Let there be 7 hyperplanes $H1, \cdots, H7$ reachable from a plane $P1$ in
$\mathbb{PG}(5, \mathbb{GF}(2))$. Let there be any other plane $P2$, which
may or may not intersect with the point set of plane $P1$. Then, any point on
$P2$ which is not reachable from plane $P1$, can \underline{maximally}
be reachable from 3 of these 7 hyperplanes, and vice-versa. Further, these 3
hyperplanes cannot be independent. Dually, any
hyperplane containing $P2$, and not containing $P1$, can \underline{maximally}
be reachable from 3 of the 7 points contained in $P1$ and which are not
independent, and vice-versa.
\end{lem}

\begin{proof}
If a point on plane $P2$ which is not reachable from plane $P1$ lies on 4
or more hyperplanes(out of 7) reachable from plane $P1$, then by lemma
\ref{4ind}, we can always find a subset of 3 independent hyperplanes in
this set of 4. In which case, the point will also be reachable
from linear
combination of these 3 independent hyperplanes, and hence to all the 7
hyperplanes which lie on plane 1. This contradicts the assumption that
the point under consideration does not lie on plane $P1$. The role of
planes $P1$ and $P2$ can be interchanged, as well as roles of points and
hyperplanes, to prove the remaining alternate propositions.
~
Hence if the point considered above lies on 3 hyperplanes reachable from
$P1$, then these 3 hyperplanes cannot be independent, following the same
argument as above. Otherwise it is indeed possible for such a point to lie
on 3 hyperplanes, e.g. in the case of the planes $P1$ and $P2$ being
disjoint.
\end{proof}

\begin{lem}
Let there be 7 hyperplanes $H1, \cdots, H7$ reachable from a plane $P1$ in
$\mathbb{PG}(5, \mathbb{GF}(2))$. Further, let there be any other plane
$P2$, which intersects $P1$ in a line. Then, there exist 4 hyperplanes
reachable from plane $P1$ which do not contain any of the 4 points that are
in plane $P2$, but not in plane $P1$.
\end{lem}

\begin{proof}
A line contains 3 points in $\mathbb{PG}(5,\mathbb{GF}(2))$. Hence,
$P1$ and $P2$ intersect in 3 points. By duality arguments, they intersect in
3 hyperplanes as well. Hence, $P2$ contains 7-3 = 4 points that are not
common to point set of $P1$. By lemma \ref{3overlap}, these 4 points can at
maximum lie on 3 hyperplanes reachable from plane $P1$. Since there are 3
hyperplanes common to  $P1$ and $P2$, and hence these 4 points already lie
on them, they do not further lie on any more hyperplane reachable from $P1$,
but not from $P2$.
\end{proof}

\begin{lem}
\label{ptoverlap}
Let there be 7 hyperplanes $H1, \cdots, H7$ reachable from a plane $P1$ in
$\mathbb{PG}(5, \mathbb{GF}(2))$. Further, let there be any other plane
$P2$, which intersects $P1$ in a point. By lemma \ref{3overlap}, the
6 hyperplanes not containing both $P1$ and $P2$ still intersect $P2$
\underline{maximum} in a line each. Then, (a) Such lines contain the common
point to $P1$ and $P2$, and hence exactly 2 more out of remaining 6 points
of $P2$ that are not common to $P1$, and (b) 3 pairs of hyperplanes out of
the 6 non-common hyperplanes intersect in a (distinct) line each out of
the 3 possible lines in $P2$ containing the common point.
\end{lem}

\begin{proof}
Let $A_{c}$ be the common atom(point) between planes $P1$ and
$P2$. By duality, exactly one hyperplane will be common to both $P1$ and
$P2$. Let some non-common hyperplane $H_{nc}$ reachable from plane $P1$
intersect plane $P2$ in a line $L1$, that is, $H_{nc} \cap P2 = L1$. Then, $A_{c}
\in L1$. For if it is not, then \\
$\mbox{ }$ $\left|L1 \cap A_{c} \right| = 4$ \\
$\mbox{ }$ $\mbox{Also, }L1 \cup A_{c} \subseteq H_{nc}$ \\
$\mbox{ }$ $\mbox{and, }L1 \cup A_{c} \subseteq P2$ \\
$\mbox{ }$ $\Rightarrow \left| H_{nc} \cap P2 \right| = 4$, a
contradiction to lemma \ref{3overlap}.
Hence, the line $L1$ contains common point $A_{c}$ and 2 more out of
remaining 6 points of $P2$ that are not common to $P1$.
~
\vspace{3pt}\\
Exactly 3 hyperplanes of the nature $H1$, $H2$, $H1+H2$ intersect in a
4-dimensional subspace, in $\mathbb{PG}(5,\mathbb{GF}(2))$.
Such a subspace can always be formed by
taking union of a plane and a line intersecting the plane in a
point, by rank arguments. Let the common hyperplane to $P1$ and $P2$ be
$H_{c}$. Let other hyperplanes reachable from $P1$ be $H1, H2, H1+H_{c},
H2+H_{c}, H1+H2, H1+H2+H_{c}$. Then, the pairs of hyperplanes $\langle$$H1,
H1+H_{c}$$\rangle$, $\langle$$H2, H2+H_{c}$$\rangle$, and $\langle$$H1+H2, H1+H2+H_{c}$$\rangle$, along with $H_{c}$,
form 3 distinct 4-d subspaces, which leads to 3 distinct lines of meet
under plane $P2$, for each pair. This can also be verified from the fact
that there are exactly 3 distinct lines in plane $P2$ that have a point
$A_{c}$ in common. These 3 lines, and their individual unions with $P1$,
leads to reachability from $H1, H_{c}, H1+H_{c}$, $H2, H_{c}, H2+H_{c}$, and
$H1+H2, H_{c}, H1+H2+H_{c}$, respectively.
\end{proof}

\begin{lem}
\label{lem5}
Let there be two hyperplanes $H1$ and $H2$ meeting in a plane $P1$. Both
$H1$ and $H2$ intersect any plane $P2$ disjoint from $P1$ in exactly a
line, by lemma \ref{3overlap}. Then these intersecting lines $L1$(of $H1$ and
$P2$) and $L2$(of $H2$ and $P2$) cannot be the same.
\end{lem}
\begin{proof}
Let the vector space of a projective geometry flat $X$ be denoted by
$V(X)$. Flats are sets of points, each of which {\it bijectively}
corresponds to a 1-d vector in the corresponding vector space. Also,
closure of a flat(in terms of containing a point) is defined as
corresponding closure of the vector subspace. Hence, family of
substructures in a projective space is bijectively intertwined to the
family of subspaces in the corresponding vector space. Then, if $L1$ = $L2$ were true, then
\begin{equation}
\label{eqn1}
V(L1)\;=\;V(L2)
\end{equation}
where
\begin{eqnarray}
V(L1)\;=\;V(H1)\cap V(P2) \\
V(L2)\;=\;V(H2)\cap V(P2)
\end{eqnarray}
Also, it is given that
\begin{eqnarray}
V(H1)\cap V(H2)\;=\;V(P1) \\
V(P1)\cap V(P2)\;=\;\phi
\end{eqnarray}
Hence if one takes closure of set of vectors contained in $V(L1)\cup V(P1)$
($L1$ is part of $P2$ which does not intersect with $P1$), it does generate
the entire vector subspace $V(H1)$. Similarly, closure of set of vectors
contained in $V(L2)\cup V(P1)$ generates the entire vector subspace
$V(H2)$. Then from equation \ref{eqn1}, the generated subspaces $V(H1)$ and
$V(H2)$ coincide, a contradiction.
\end{proof}

\subsection{Lemmas Related to Embedded Graphs}
\begin{lem}
\label{9_8}
In a bipartite graph having 9 vertices each in both partite sets, and having
a minimum degree($\delta$) of at least 8, any collection of 3 vertices from
one side is incident on at least 6 common vertices on the other side.
\end{lem}
\begin{proof}
Let the vertices of one side be denoted as (a1, a2, $\cdots$, a9), and that
of other side by (b1, b2, $\cdots$, b9). Given $\delta$ = 8, it is obvious
that minimal intersection of neighborhoods of a1 and a2 happens in some(at
least) 7 vertices from the other side. Then the 2 remaining vertices are
$N(a1) - N(a1)\cap N(a2)$ and $N(a2) - N(a1)\cap N(a2)$, respectively. The
neighborhood of vertex a3 may either contain all these 7 vertices
($N(a1)\cap N(a2)$), or the two vertices $N(a1) - N(a1)\cap N(a2)$ and
$N(a2) - N(a1)\cap N(a2)$, and at least 6 vertices out of $N(a1)\cap N(a2)$.
Hence the minimal intersection of neighborhoods of arbitrarily chosen
vertices a1, a2 and a3 is of cardinality 6.
\end{proof}

\begin{lem}
\label{10_8}
In a bipartite graph having 10 vertices each in both partite sets, and having
a minimum degree($\delta$) of at least 8, any collection of 3 vertices from
one side is incident on at least 4 common vertices on the other side.
\end{lem}
\begin{proof}
Let the vertices of one side be denoted as (a1, a2, $\cdots$, a9), and that
of other side by (b1, b2, $\cdots$, b9). Given $\delta$ = 8, it is obvious
that minimal intersection of neighborhoods of a1 and a2 happens in some(at
least) 6 vertices from the other side. The two vertices in $N(a1) -
N(a1)\cap N(a2)$ and two more in $N(a2) - N(a1)\cap N(a2)$ count the 4
remaining vertices on the other side. The
neighborhood of vertex a3 may either contain all these 6 vertices
($N(a1)\cap N(a2)$), or at most all 4 vertices $N(a1) - N(a1)\cap N(a2)$ and
$N(a2) - N(a1)\cap N(a2)$, and at least 4 vertices out of $N(a1)\cap N(a2)$.
Hence the minimal intersection of neighborhoods of arbitrarily chosen
vertices a1, a2 and a3 is of cardinality 4.
\end{proof}

\begin{lem}
\label{11_8}
In a bipartite graph having 11 vertices each in both partite sets, and having
a minimum degree($\delta$) of at least 8, any collection of 3 vertices from
one side is incident on at least 2 common vertices on the other side.
\end{lem}
\begin{proof}
Let the vertices of one side be denoted as (a1, a2, $\cdots$, a9), and that
of other side by (b1, b2, $\cdots$, b9). Given $\delta$ = 8, it is obvious
that minimal intersection of neighborhoods of a1 and a2 happens in some(at
least) 5 vertices from the other side. The three vertices in $N(a1) -
N(a1)\cap N(a2)$ and three more in $N(a2) - N(a1)\cap N(a2)$ count the 6
remaining vertices on the other side. The
neighborhood of vertex a3 may either contain all these 5 vertices
($N(a1)\cap N(a2)$), or at most all 6 vertices $N(a1) - N(a1)\cap N(a2)$
and  $N(a2) - N(a1)\cap N(a2)$, and at least 2 vertices out of
$N(a1)\cap N(a2)$. Hence the minimal intersection of neighborhoods of
arbitrarily chosen vertices a1, a2 and a3 is of cardinality 2.
\end{proof}

\section{Performance of Code for Burst Errors}
\label{burst_perf}

The strongest applications for this code lie in the areas of mass data
storage such as discs. Here, as pointed earlier in section \ref{intro_sec},
burst errors are the dominant cause of data corruption. Hence we have also
examined the burst error correction capabilities of our code.

In bipartite graph $\mathbb{G}$ constructed from
$\mathbb{P}(5,\mathbb{GF}(2))$, we label the edges with integers, to
map various symbols of the codeword. Such a labeling is not required
to understand/ characterize the random error correction capability of the
code. But here, we label the edges with numbers to try to maximize the burst
error correction capability. This is achieved if each consecutive symbol,
possibly part of a burst, is mapped to edges that are incident on distinct
vertices representing different component decoders. Thus, consecutive
numbered edges, representing consecutively located symbols in input symbol
stream, go to different vertices hosting different RS decoders. Since
there are 63 vertices on one side of the graph, this scheme of numbering
implies that edges incident on vertex 1 are assigned the numbers \{1, 64,
\ldots, 1890\}. Similarly, the edges incident on Vertex 2 are assigned
\{2, 65, \ldots, 1891\}, and so on. This numbering essentially achieves the
effect of \textit{interleaving} of code symbols. If the error correcting
capacity of each component RS decoder is $\mu$(=
$\lfloor\frac{\epsilon}{2}\rfloor$), then the minimum burst error
correcting capacity of $\mathbb{C}$ will be $\mu*63$. For example, for
$\epsilon=5$, $\mu$ is 2, and the minimum burst error correcting capacity
is $2*63=126$. Table \ref{tab:4} gives MATLAB simulation results for
burst error correction for $\epsilon=5$.

\begin{table}[h]
\centering
\begin{tabular}{|c|c|c|}
\hline
No. of errors & failures & Avg. no. of iterations\\ \hline
126 & 0 & 1 \\ \hline
135 & 26 & 2.43 \\ \hline
\end{tabular}
\caption{Burst errors ($\epsilon=5$)}
\label{tab:4}
\end{table}
To demonstrate the \textit{excellent burst error correction capacity} of
our code, we benchmark it against the massive interleaving based codes in
CD-ROMs.  These codes are considered to be very robust to burst errors.
Traditionally in ECMA-130, the encoding utilizes heavy interleaving and
dependence on erasure correction to deal with burst errors. For erasure
correction, one level of decoding identifies the possible locations of the
error symbols. The next level of decoding uses this information to correct
them. The stage/process of interleaving used in CD-ROMs makes the encoding
and decoding slower. We propose two schemes in \cite{swad_rep}, which offer
significant improvement in burst error correction at similar data rates.
Our decoder, being fully parallel in its decoding, can handle larger sets
of data at a time and hence could be used to increase the throughput. In
our schemes, however, we wanted to fit our decoders in place of the heavy
interleaving stage of the CD-ROM decoding data path, which only processes
one frame at a time. Thus, in terms of throughput we will be matching the
CD-ROMs but we will surpass them in burst error correction capability.

\section{A Note on Encoding}
For the encoding process, we derive the parity matrix and find its orthogonal
matrix to get the generator matrix.
Suppose $d=5$ which means that for each sub-code 4 edges act as parity symbols.
The parity matrix for each vertex is given by:
\[\left(\begin{array}{ccccc}
1 & \alpha & \alpha^2 & \ldots & \alpha^{30} \\
1 & \alpha^2 & \ldots & \ldots & \alpha^{60} \\
. & \ldots & \ldots & \ldots & . \\
1 & \alpha^4 & \ldots & \ldots & \alpha^{120} \end{array} \right )\]

where $\alpha=2$ for us. The parity matrix for the 126 vertices is generated
using the above parity matrix and inserting the appropriate entries at the
corresponding edge locations. Each column of the overall parity matrix
corresponds to an edge and there are 126*4 rows.
We then perform row operations to get it in RRE
form. The generator matrix is then easily obtained. A codeword is given by the
product of the message vector with the generator matrix.

The above stated method is not efficient because it uses matrix-vector
multiplication and hence if of $O(n^2)$. More efficient methods could be possible 
by utilizing the structure of the graph. Getting an efficient encoder design is
one of the possibilities of future work in this area.

\section{Results}
\label{results_sec}
For proof of concept, we have done VLSI prototyping of efficient,
throughput-optimal decoder for the expander-like code as well. The code
used was the length-1953 code. (31, 25, 7)
Reed-Solomon codes were chosen as subcodes, and (63 point, 63 hyperplane)
bipartite graph from {\large $\mathbb{P}(5,\mathbb{GF}(2))$} was chosen as
the \textit{expander graph}. The overall expander code was thus (1953,
1197, 761)-code. 

Viewed as a \textit{computation
graph}, every vertex of this graph maps to a \textit{RS decoding
computation}. The input symbols to each of these decoders correspond to
the edges which are incident to the vertex in question. The prototype
implementation was done on a Xilinx XUPV506 board based on LX110T
FPGA, with speed grade of -3. It
uses the RS decoder IP provided by Xilinx itself. The decoder can be made
to skip decoding, to accommodate the modification to Zemor's algorithm done
by us. We could also use it to perform erasure correction, since erasures
arise frequently in secondary storage device data. Using projective
geometry properties, we evolved a \textbf{novel}, throughput-optimal
strategy to \textbf{fold} the parallel vertex computations, such that the
number of RS decoders required to implement the decoder for $\mathbb{C}$ is
only a factor of order of $\mathbb{G}$. This saves a lot of resources, and
can fit on even small FPGAs. The entire folding strategy has been detailed
in \cite{dmaa_pap}.

For such (folded) design, about 25\% of the FPGA slices
were used to implement the decoder.
We used distributed RAM to implement the
memory modules.The post placement-and-routing frequency was found to be
180.83 MHz for the design without erasures, and 180.79 MHz for the
design with erasures. The design was based on $\epsilon=5$, for which it
takes 2611 clock cycles to finish 4 iterations, without erasure decoding.
Adding 217 clock cycles to write data into the memory, we got a throughput
of $\frac{1953}{2828}*181\approx 125Mbytes/s$.
More complete details of this implementation, and its
performance figures such as throughput, can be found in \cite{swad_rep}.
The efficient design of the decoder is patent pending
\cite{expander_patent}. The error-correction performance by {simulating this
VLSI model}{testing this implementation working on FPGA board} is tabulated
as following.

\begin{table}[h]
\centering
\begin{tabular}{|c|c|c|c|c|}
\hline
$\epsilon$ & Latency & Processing Delay & Random errors & Burst errors \\ \hline
5 & 83 & 45 & 141 & 143 \\ \hline
7 & 115 & 77 & 218 & 219 \\ \hline
9 & 155 & 117 & 328 & 295 \\ \hline
\end{tabular}
\caption{Erasure correction results}
\label{erasuretab}
\end{table}

\section{Applications to Storage Media}
\textit{Disc storage} is a general category of secondary
storage mechanisms. Unlike the
now-obsolete 3.5-inch floppy disk, most removable media such as optical discs
do not have an integrated protective casing. Hence they are susceptible to
data transfer problems due to scratches, fingerprints, and other environmental
problems such as dust speckles. These data transfer problems, while the data
is being read, manifests itself in form of bit errors in the digital data
stream. A long sequence of
bit read errors while a track is being read (e.g. a scratch on a track)
can be characterized as \textit{burst error}, while bit read error arising
out of tiny dust speckle masking limited number of pits and lands on a track
leads to \textit{random error}.  The occurrence of such events obviously not
being rare, \textit{recovery of data to maximum extent} in presence of such
errors is an \textit{essential} subsystem within most computing systems,
such as CPU and disc players.

\subsection{Application to CD-ROM}
\label{cd_app}
As indicated in section \ref{burst_perf}, we have also benchmarked the
burst error correcting capability of our code against CD-ROM decoding based
on ECMA-130 \cite{ECMA}. Based on this benchmarking, we have proposed 2
novel schemes for CD-ROM encoding and decoding stages. These schemes are
based on the expander-like codes described in this paper. Application of
these codes at various stages of CD-ROM
encoding scheme(and correspondingly in decoding scheme)
\textit{substantially} increases the burst error correcting capability of
the disc drive.

To start with, we \textit{recall} from \cite{ECMA} that the major part of
error correction of the CD-ROM coding scheme occurs during the two stages,
RSPC and CIRC. On the encoder side, RSPC stage comes before CIRC stage, while on decoder
side, CIRC stage comes before RSPC stage. To get an idea of the average
error correction capabilities of CDROM scheme, we simulated the CIRC and
the RSPC stages of the ECMA standard in MATLAB.
The details of these stages are as follows.

\begin{description}
\item[CIRC]
This stage leads to interleaving of codeword symbols. The massive
interleaving done here is mainly responsible for the burst error
correction. In a frame of 6976(=32*109*2) symbols, it can maximum correct
480 symbols of burst errors.

\item[RSPC]
After the CIRC stage during decoding, the remaining errors can be considered as \textit{random errors}. The RSPC
stage in decoding then serves to correct these errors using RS decoding as
erasure decoding. If we consider only error
detection and correction, the CIRC + RSPC system \textit{on an average} corrects a
burst of 270 symbols in a frame of 6976 symbols.

\end{description}

The schemes we propose involve replacing one or both of the CIRC and the
RSPC with encoders and decoders based on our expander-like codes. This
happens in such a way, that we maximize burst error correction, without
suffering too much on the data rate.

\subsubsection{Scheme 1}

In the first scheme, we replace CIRC+RSPC subsystem with \textbf{4
decoders} for
our expander-like codes, $\mathbb{C}$. Hence the RS subcodes used in
$\mathbb{C}$ have block length $d$ as 31(symbols). Further, we fix their
minimum distance as $\epsilon=7$, which also implies that their data rate
is $\frac{k}{n}=\frac{25}{31}=0.806$. The output of corresponding 4
encoders is further interleaved to improve performance, and de-interleaved
on receiving side. 

To construct these subcodes, we take a (255,249,7) RS code, and
\textit{shorten} it by using the first 31 8-bit symbols only. For the
overall code $\mathbb{C}$, the data rate is equal to
($2*r-1$)\cite{sipser}, where $r$ is the rate of the (RS) subcode. Hence the
rate for codes used in each of our encoders/decoders is $2*0.806-1$ =
$0.612$. Thus, the number of message symbols for each decoder is equal to
$1953*0.612$ = $1197$. The rest are therefore parity symbols.

In terms of frames, we set the cumulative input of the encoders, and
correspondingly the cumulative output of decoders, to a stream of 199
frames, each having 24 symbols payload. Assuming that each symbol can be
encoded in 1 byte, this leads to generation of 4776 bytes. With 12
padding bytes added to it, we can re-partition this bigger set of 4788
bytes into 4 blocks of 1197 bytes each(4*1197=4788). Each block of 1197
source symbols can then be worked upon by 4 \textit{parallely} working
encoders. After encoding each block to 1953 symbols, one of the extra
added(padding) bytes is removed from each encoder giving 244 frames of
32 bytes of data. Every RS decoder has $\epsilon=7$, which implies that
it can detect and correct upto 3 errors. Thus, each encoder for
$\mathbb{C}$ will give a burst error correcting capability of
$63*3 = 189$. Since 4 of such encoders work in interleaved fashion, we
get a minimum burst error detection and correction capability of
\textbf{756 among 244 frames}. This is \textit{opposed to 270 in 218
frames}, in the case of CIRC+RSPC subsystem.
One can hence clearly see the massive improvement in burst error
correction, at a comparable code rate(0.62 for our code vs. 0.75 for
CIRC+RSPC).  One disadvantage of this scheme is it being hardware expensive
due to use of many parallel RS decoders. Also, the high throughput of our
decoder is not utilized. We are limited by
the stages before and after the CIRC+RSPC subsystems(see \cite{ECMA}).
Hence, even though our decoder is faster, the advantage is not seen.

\subsubsection{Scheme 2}
\label{scheme2}

This scheme is a more hardware economical scheme, which also increases the
burst error correcting capability. Since our decoder also has a very good
random error correcting capability, we can achieve an error correction
advantage by replacing just the RSPC stage with our encoder/decoder.
\textbf{Two}
of our encoders can take the place of the RSPC encoder in this scheme. Data
from these encoders is then interleaved, and passed on to the CIRC. In the
decoding stage, after CIRC, there is correspondingly de-interleaving
followed by decoding based on our code.

This scheme has the advantage that it increases the error correction
capability. It also matches the code rate of CIRC: 0.75 for CIRC, versus
0.74 for our decoder. Also, it is a hardware economical scheme. MATLAB
simulations show that the burst error rate goes up from 270 for CIRC+RSPC
subsystem, to more than 400 for CIRC and our encoder. Tables \ref{tab:5}
and \ref{tab:6} show some simulation results for this scheme.

\begin{table}[h]
\centering
\begin{tabular}{|c|c|}
\hline
No. of errors & failures \\ \hline
270 & 2 \\ \hline
300 & 45 \\ \hline
400 & 86 \\ \hline
\end{tabular}
\caption{Response to Burst errors for CIRC+RSPC}
\label{tab:5}
\end{table}
\vspace{-.1in}

\begin{table}[h]
\centering
\begin{tabular}{|c|c|}
\hline
No. of errors & failures \\ \hline
400 & 7 \\ \hline
450 & 17 \\ \hline
500 & 26 \\ \hline
\end{tabular}
\caption{Response to Burst errors in Scheme 2}
\label{tab:6}
\end{table}
\vspace{-.1in}

\subsubsection{Meeting Throughput Requirement}
For a 72x CD-ROM read system, the required data transfer rate is
$10.8Mbytes/s$. Recall from section \ref{results_sec} that the decoder for
our codes achieved a throughput of $\approx$ 125 $Mbytes/s$. Hence, this
decoder using our codes can easily be incorporated without hurting
throughput. Moreover, in an ASIC implementation, we would expect the
performance to be better.

\subsection{Application to DVD-R}
\label{dvd_sec}
The same class of codes can also be applied to evolve encoding and decoding
for DVD-ROM. The main error correction in DVD-R is provided by
the RSPC block \cite{dvdr}, which consists of an inner RS(182,172,11) code and an
outer RS(208,192,17) code. The inner code can detect and correct up to
5 errors, while the outer code can detect and correct 8
errors. A detailed analysis shows \cite{swad_rep} that without considering
erasure decoding, we can still detect and correct $2922$ errors
in a burst, in one block of 37856 bytes. If we allow for the inner decoding
to mark as \textbf{erasures}, $5834$ bytes errors can be corrected.

As an alternate decoding scheme, we substitute the RSPC stage of the DVD encoding by encoders
of code $\mathbb{C}$. These encoders are therefore employed during the
transformation of \textit{Data Frames} into \textit{Recording Frames}. We
can replace the RSPC stage with a expander-like code encoder/decoder, made
from point-hyperplane graph of $\mathbb{PG}(8,GF(2))$, and component
RS(255,239,17) codes. The overall burst error correction capability
\textit{without erasures} turns out to be $8*511=4088$ bytes, which is much
greater than 2922. As far as random errors are concerned, MATLAB simulation
results show that around 1990 random errors are always corrected in one
iteration of the decoding itself. The existing standard specifies that
the number of random errors in 8 consecutive ECC blocks must be less than
or equal to 280.  The complete details of this scheme can be found in
\cite{swad_rep}.

This particular application of our codes also brings out the fact that
taking a bipartite graph $\mathbb{G}$ from a \textit{higher-dimensional}
projective space can be advantageous in terms of better rate and better
error correction capacity.

\section{Conclusion}

We have presented the construction and performance analysis for an
expander-like code that is based on a bipartite graph. The
graph, derived from the incidence relations of projective spaces, offers unique
advantages in terms of deriving lower bounds on error correction capabilities.
There are also fundamental advantages in terms of hardware design of
decoder for this code. As the size of the graph increases, practical
implementation of the decoder becomes difficult. Projective geometry,
through lattice embedding properties, offers a natural way of efficient
folding the computations which leads to using fewer processors, while
guaranteeing throughput-optimality \cite{dmaa_pap}.

The error-correction performance of our codes is better than previously
stated in literature. This is because it relaxes some restrictions that
were imposed with respect to the second largest eigenvalue of the graph.
Derivation of bounds of error correction have been presented,
and the \textit{average case performance} of the code is shown to be up to 10 times
better through simulations. Moreover, the code has special implicit interleaving
due to the numbering of the edges. This offers great advantage in burst error
corrections. A natural application of these codes with respect to data storage
media (namely CD-ROMs and DVD-R) has been explored. We have presented schemes
that improve the burst error performance in comparison to existing standards. 

The hardware design for the decoder has been completely worked out. We use
Xilinx RS decoder IPs as the processors. The computations have been efficiently
folded in order to make them fit on a Xilinx LX110T FPGA. We have tested the
design on the FPGA and also exploited the erasure correction ability of the
RS code. Moreover, a general folding strategy has been developed for higher
dimensions of projective geometry that provides a methodology for practically
implementing decoders of higher dimensions. 

Overall, we believe that further research should establish even more
usefulness of our expander-like codes.

%
%
%
%



%

\newpage

%
%
%
%
\begin{appendix}
\vspace{-.1in}
\section{Proof for Random Error Capability in case of $\epsilon$=15}
\label{proof_app}
Before going into details of the proof, we need to establish certain
\textbf{cardinalities} related to $\mathbb{PG}$(5,$\mathbb{GF}$(2)).
From the discussion in section \ref{pg_intro}, we get the following.
\begin{enumerate}
\item Any line in this space is defined by any 2 points. The unique third
point lying on the lines is determined by linear combination of
corresponding one-dimensional subspaces. Hereafter, a line will be
represented as a tuple $\langle$a, b, a+b$\rangle$.
\item Similarly, or dually, 3 hyperplanes intersect in a particular 4-d
projective subspace, or flat.
\item Any plane in this space is defined by 3 non-collinear points, and
their 4 unique linear dependencies in the corresponding vector space.
Hereafter, a plane will be represented as $\langle$a, b, c, a+b, b+c, a+c,
a+b+c$\rangle$, with the non-canonical choice of non-collinear points within this
plane being implicit as $\langle$a, b, c$\rangle$.
\item A plane contains 7 lines, thus being a Fano plane.
\item A plane is reachable from 7 hyperplanes and 7 points in the
corresponding lattice structure through transitive join and meet
operations. Such an hourglass structure will be critical in our proofs.
\item Similarly, a hyperplane is reachable from 31 points: 5 of them being
independent, and others representing the linearly dependent vectors on these.
\end{enumerate}

\subsection{Main Proofs}
Proving the two theorems of section \ref{props} is done by
demonstrating how one can incrementally construct an embedded bipartite
subgraph, by improving over minimum degree of a smaller subgraph. This requires
looking at the planes involved in the construction of the embedding.

\begin{theor}
In the construction of bipartite graph mentioned in section \ref{code_sec},
there exists no embedded subgraph having size of partitions as 9, the
degree of each of whose vertices has a minimum degree($\delta$) of 8.
\end{theor}
\begin{proof}
Assume that such a subgraph exists. Then by lemma \ref{9_8}, any 3
points have at least 6 hyperplanes in common, and vice-versa. By lemma
\ref{4ind}, the set of 9 points contains at least 3 non-collinear points. The
6 common hyperplanes to them in the subgraph must contain the (unique) plane
defined by the points. If one of the points is contained in some different plane,
then by lemma \ref{3overlap}, this point can only be reachable from maximum 3
hyperplanes belonging to the first plane, and not (all) 6 common hyperplanes, a
contradiction. Again, from lemma \ref{4ind}, one can pick a subset of 3
independent hyperplanes out of these 6 common hyperplanes. By dual
arguments, these 3 hyperplanes also must have 6 points in common, reachable
from a single unique plane defined by the 3 hyperplanes. Thus, a 6-a-side
bipartite subgraph with reachability defined by a single plane of the
underlying projective space, is embedded in the 9-a-side bipartite subgraph
we are trying to construct.
\begin{list}{\textbullet}{\leftmargin=17pt}
\item \textit{Case 1:} Out of the remaining 3 points(hyperplanes) in the
9-a-side subgraph, we can at most choose 1 point such that it is contained
in all the 6 hyperplanes. This 1 point is the $7^{th}$ point on the 7-7
hourglass passing through a single plane. It also involves the remaining
$7^{th}$ hyperplane being incident on all these 6+1 points. The remaining
2 points are necessarily part of at least one other plane. This
configuration of 2 remaining points and 2 remaining hyperplanes may
maximally form a $K_{2,2}$ complete induced subgraph by considering
whichever number of intervening planes. In terms of their minimum degree,
these 2 remaining points, by lemma {3overlap}, can at maximum be reachable from 3
more hyperplanes reachable from plane $P1$. Hence the maximum degree these
two points achieve is 5, in any possible construction. This contradicts the
presence of assumed subgraph having $\delta$ of 8.

\item \textit{Case 2:} On similar lines, if we choose the 3
remaining points and hyperplanes apart
from the 6-a-side subgraph to form a complete bipartite subgraph $K_{3,3}$ by
any construction, then again by lemma \ref{3overlap}, they can at maximum be reachable from 3 more hyperplanes, reachable from plane $P1$. In such a case,
they maximally achieve a minimum degree of 6, which is still lesser than 
requirement of 8.
\end{list}
\end{proof}

\begin{theor}
In the construction of bipartite graph mentioned in section \ref{code_sec},
there exists no embedded subgraph having size of partitions as 10, the
degree of each of whose vertices has a minimum degree($\delta$) of 8.
\end{theor}

\begin{proof}
Assume that such a subgraph exists. Then by lemma \ref{10_8}, any 3
points have at least 4 hyperplanes in common, and vice-versa. By lemma
\ref{4ind}, the set of 10 points contains at least 3 non-collinear points. The
4 common hyperplanes to them in the subgraph must contain the (unique) plane
defined by the points. If one of the points is contained in some different plane,
then by lemma \ref{3overlap}, this point can only be
reachable from maximum 3
hyperplanes belonging to the first plane, and not (all) 4 common hyperplanes, a
contradiction. Again, from lemma \ref{4ind}, one can pick a subset of 3
independent hyperplanes out of these 4 common hyperplanes. By dual
arguments, these 3 hyperplanes also must have 4 points in common, reachable
from a single unique plane defined by the 3 hyperplanes. Thus, a 4-a-side
bipartite subgraph with reachability defined by a single plane of the
underlying projective space, is embedded in the 10-a-side bipartite subgraph
we are trying to construct.
~
\vspace{3pt}\\
By considering only one intervening plane, we can maximum go upto 7-a-side
subgraph only. Hence to construct 10-a-side graph, we need to consider at
least one more plane in the construction. We now individually consider
the cases where the maximum number of points taken from any one of the
planes is n: $4\leq n \leq 7$.
\begin{list}{\textbullet}{\leftmargin=17pt}
\item \textit{Case 1:} Assume that the maximally possible set of 7 points
and some number of hyperplanes are taken from the plane $P1$ intervening
the 4-a-side subgraph already present. The number of hyperplanes considered
from $P1$ could therefore be anywhere between 4 and 7. Hence we need to
consider at least one more intervening plane between the remaining
hyperplanes(minimum: 3, maximum: 6) and the 3 remaining points. In the best
possible construction, these remaining hyperplanes and points form a
biclique. Then, any particular hyperplane from this biclique is reachable
from a maximum of all 4 points of this biclique, plus at maximum 3 more
points of plane $P1$; refer lemma \ref{3overlap}. Hence the degree
requirement of these hyperplanes is unsatisfiable using a 7-* partition of
the 10-point set, a contradiction.

\item \textit{Case 2:} Next, assume that 6 points and some number of
hyperplanes are taken from the plane $P1$ intervening the 4-a-side subgraph
already present. The number of hyperplanes considered from $P1$ could
therefore be anywhere between 4 and 7. Hence again we need to consider at
least one more intervening plane between the remaining hyperplanes(minimum:
3, maximum: 6) and the 4 remaining points. In another best
possible construction, these remaining hyperplanes and points form a
biclique. Then, any particular hyperplane from this biclique is reachable
from a maximum of all 4 points of this biclique, plus at maximum 3 more
points of plane $P1$; refer lemma \ref{3overlap}. Hence the degree
requirement of these hyperplanes is unsatisfiable using a 6-* partition of
the 10-point set, a contradiction.

\item \textit{Case 3:} Next, assume that 5 points and some number of
hyperplanes are taken from the plane $P1$ intervening the 4-a-side subgraph
already present. The number of hyperplanes considered from $P1$ could
therefore be anywhere between 5 and 7. Hence again we need to consider at
least one more intervening plane between the remaining hyperplanes(minimum:
3, maximum: 5) and the 5 remaining points.
~
\vspace{3pt}\\
First, we claim that under this case, a subgraph $K_{5,5}$ having one
intervening plane always exists. To see that, let's take the lone boundary
case where 5 points and 4(minimum) hyperplanes are taken from plane $P1$,
and hence a $K_{5,5}$ biclique is not provided by incidences of $P1$.
Then, the remaining 6 hyperplanes must belong to plane/s different from
$P1$. By lemma \ref{3overlap}, they can at maximum be reachable from 3 of the 5
points reachable from $P1$. To satisfy their requirement $\delta \geq 8$,
they should therefore be reachable to all 5 of the remaining points. Hence
the set of 6 remaining hyperplanes and 5 remaining points form a $K_{6,5}$
biclique, and hence also $K_{5,5}$. By lemma \ref{4ind} and the fact that
a plane formed by 3 independent points is reachable from 7 hyperplanes, which
is simultaneously minimum and maximum, this $K_{5,5}$ biclique contains
exactly 1 intervening plane different from $P1$.
~
\vspace{3pt}\\
By abuse of notation, let's call the plane intervening the always-present
$K_{5,5}$ subgraph as $P1$. Then, at maximum 7 hyperplanes can be
considered from $P1$ in the construction, and hence 3, 4 or 5 hyperplanes
and remaining 5 points need to be added to $K_{5,5}$ by considering other
planes. In case when either 3 or 4 hyperplanes are considered from other
planes, the set of 5 remaining points cannot have their degree requirements
satisfied. For, these points can be reachable from maximum 4 hyperplanes from
other planes, and maximum 3 more, considering plane $P1$(refer lemma
\ref{3overlap}). Hence we look into the case when 5 hyperplanes, and 5
points are considered by looking into other planes.
~
\vspace{3pt}\\
In this case, each hyperplane out of 5 remaining hyperplanes needs to be
reachable from 3 different points reachable from $P1$. These 3 different
points cannot be independent, for if they were, then the corresponding
hyperplane would have been on $P1$ rather than any other plane. Hence each
one out of 5 such collections of 3 points forms a line. Given a plane in
$\mathbb{PG}(5,\mathbb{GF}(2))$ having its point set as
$\langle$a,b,c,a+b,b+c,a+c,a+b+c$\rangle$, it is immediately obvious that no subset of
5 points contains 5 different lines. In fact, to have 5 different lines
contained in some point subset, the minimum size of the subset required is
7. Hence all possible constructions in this case leaves at least one
hyperplane not having its degree requirement satisfied, a contradiction.

\item \textit{Case 4:} Finally, assume that 4 points and some number of
hyperplanes are taken from the plane $P1$ intervening the 4-a-side subgraph
already present. The number of hyperplanes considered from $P1$ could
therefore be anywhere between 4 and 7. Hence again we need to consider at
least one more intervening plane between the remaining hyperplanes(minimum:
3, maximum: 6) and the 6 remaining points. We consider following two
cases.
\begin{list}{$\diamond$}{\leftmargin=18pt}
\item In this case, we assume that the remaining 6 points contain at least
one subset of size 3 forming a line. At least one point out of 3 remaining
points of this 6-set will be not be part of this line(a line has maximum 3
points in $\mathbb{PG}(5,\mathbb{GF}(2))$). This point and the line
therefore form a plane $P2$, which is maximal, by our assumption in Case 4.
Since in this case, \textit{any 3} points will have at least 4 common
hyperplanes to satisfy their degree requirements, the 3 independent points
of plane $P2$ will have 4 hyperplanes in common, and vice-versa. The
remaining 2 hyperplanes, hereafter referred as $H1$ and $H2$, do not
contain both $P1$ and $P2$. More formally, by lemmas \ref{ptoverlap} and
\ref{lem5}, the best case occurs when \\
~
$\mbox{ }$ $H\:i \cap P\:i\;=\;\mbox{a line, for i, j=1 and 2}$ \\
Hence these hyperplanes can be reachable from a maximum of 6 points reachable
from $P1$ and $P2$. In fact, they need to be
reachable from exactly 6 to satisfy
their degree requirements. This reachability from 6 points by each of the 2
hyperplanes must consist of
reachability from one line each from planes $P1$ and $P2$.
~
\vspace{3pt}\\
By lemmas \ref{ptoverlap} and \ref{lem5}, the intersecting lines of $H1$
and $H2$ to say, $P1$, cannot be the same. In a plane of
$\mathbb{PG}(5,\mathbb{GF}(2))$ denoted as $\langle$a,b,c,a+b,b+c,a+c,a+b+c$\rangle$,
one can clearly see that to accommodate 2 different lines, one needs to
consider a subset of at least 5 points. This contradicts our construction
in which we originally had 2 collections of 4 points each contained in two
planes.

\item In this case, we assume that the remaining 6 points do not contain
any subset of size 3 which is dependent. This means that any subset of
triplet of points from among these define a plane. So we will arbitrarily
consider two disjoint triplets from this 6 remaining points, and consider
their planes $P2$ and $P3$. Also note that points from triplet of $P2$
do not lie on $P3$, and vice-versa. For, we are assuming in this sub-case
that 4 coplanar points do not exist among these 6 remaining points.
Again in this case, \textit{any 3} points will have at least 4 common
hyperplanes to satisfy their degree requirements, and vice-versa. Hence we
again have 4 points and hyperplanes being considered in construction, for
each of the planes $P1$, $P2$ and $P3$. Note that there may be common points and
hyperplanes used in this construction. We consider 2 separate sub-subcases
\begin{list}{$\star$}{\leftmargin=18pt}
\item In this case, all the 10 hyperplanes of the required graphs lie
within the set of union of sets of 4 hyperplanes each being considered per
plane.
~
\vspace{3pt}\\
Consider the 3 independent points of the triple of plane $P2$.
None of these points can be same, or linear combination
of any point reachable from plane $P1$. These 3 points were earlier also
illustrated to be independent, and hence not reachable from $P3$ as well.
A join of plane $P1$ and one different point reachable from plane $P2$ in
the corresponding lattice yields one different 3-d flat each, in the
lattice. With respect to plane $P1$ where we are considering a set of 4
hyperplanes, it is obvious that in $\mathbb{PG}(5,\mathbb{GF}(2))$, at
least 3 of these hyperplanes are independent. If the remaining hyperplane
is dependent on 2 of these 3, then a complete 3-d flat is reachable from
this set of 4 hyperplanes.
By lemma \ref{3overlap}, any point of $P2$ is reachable
from at most 3 of
the hyperplanes reachable from $P1$, that too when the 3 hyperplanes are
part of a complete 3-d flat. Hence it is possible that a particular
point reachable from $P2$ is also reachable
from the unique 3-d flat, whenever it exists, and thus to the 3 hyperplanes of
$P1$ from this 3-d flat. Since the join of different independent points of
$P2$ with plane $P1$ gives rise to different 3-d flats, and since there is
at most one 3-d flat embedded in the set of 4 hyperplanes being considered
w.r.t. plane $P1$, at least two points reachable from plane $P2$ can only
be reachable from at most 2 hyperplanes reachable from plane $P1$.
A similar conclusion can be reached with respect to the same 3 points of
plane $P2$, and the hyperplanes reachable from $P3$, that are under
consideration for this construction.
In the best case, 2 distinct points out of the 3 points reachable
from $P2$ have a degree of 3 towards planes $P1$ and $P3$, respectively.
Hence at least 1(the remaining one) point reachable from plane $P2$ is
reachable from at most two hyperplanes each, reachable from planes $P1$ and
$P3$. A similar point can also be located on plane $P3$, the 3 points
reachable from which have identical relation to the point sets of remaining
2 planes.
~
\vspace{3pt}\\
Without loss of generality, we further claim that at most 3 of the 10
hyperplanes used in construction remain outside of those considered
for planes $P1$ and $P3$(or $P2$). When planes $P1$ and $P3$ are
disjoint, they cover 8 of the 10 required hyperplanes of the construction.
If the planes meet in a point, then by duality arguments, they also meet in
a hyperplane. In which case, they cover 7 out of 10 required hyperplanes of
the construction. If $P1$ and $P3$ meet in a line(the last possible
scenario), then $P3$ has 1 point exclusively belonging to itself. Hence in
all cases, either $P2$ or $P3$ has at most 3 points lying on it. On both
these planes, we have located at least 1 point, which is reachable from at
most 2 points each to the remaining 2 planes. Hence in all scenarios, there
exists at least one point, which can be reachable from at most 3+2+2 = 7
hyperplanes, thus invalidating the construction of this case.

\item In this case, all the 10 hyperplanes of the required graphs do not lie
within the set of union of sets of 4 hyperplanes each being considered per
plane. Hence, there is at least 1 hyperplane that is not covered by planes
$P1$, $P2$ and $P3$, i.e. it is not reachable from either of these. By lemma
\ref{3overlap} and the fact that in the assumption for this case, the triplet
of points of $P2$ and $P3$ are not collinear, one can see that the
points of the planes $P2$ and $P3$ provide degree at most 2 each to the
above hyperplane. Additionally, it may provide degree of/be reachable from
maximum 3 points lying on plane $P1$. By considering the planes $P1$, $P2$
and $P3$, we have exhausted all the points of the construction, and the maximum
degree this particular hyperplane has achieved so far is 3+2+2 = 7,
that is clearly not sufficient.
\end{list}
\end{list}
\end{list}
\end{proof}
%
\section{An Eigenvalue-based Approach for Deriving $\xi$}
\label{eigen}
This section provides an \textit{easier, alternative} approach towards deriving
\textit{weaker} upper bounds on random error correction capability of our code.
This approach can be used for the extreme cases when one has to consider the code rate
that is less than 0.1 (equivalently, $\epsilon \geq 15$).
The arguments in this approach are very similar to the ones given by \cite{zemor}.
Let $\mathbf{A}=(\mathbf{a_{ij}})$ be the $2n$ {\tiny X} $2n$ adjacency matrix
of the bipartite graph $\mathbb{G}(V,E)$ of degree $d$.
That is, $a_{ij}=1$ if there
is an edge between the vertices indexed by $i$ and $j$, and $a_{ij}=0$
otherwise. Let $\mathbb{S}$ be the set of vertices of the graph $\mathbb{G}$
that form the minimal configuration of failure. Let $\mathbf{x_s}$ be a
column vector of size $2n$ such that every coordinate indexed by a
vertex of $\mathbb{S}$ equals 1 and the other co-ordinates equal 0.
Now, we have,
\begin{equation}
	\mathbf{x_s^TAx_s \; = \; \sum_{v\epsilon\mathbb{S}} \delta_{G_\mathbb{S}}(v)}
\end{equation}
where $\delta_{G_\mathbb{S}}(v)$ is the degree of vertex $v$ in the subgraph $G_\mathbb{S}$
induced by the vertex set $\mathbb{S}$ in $\mathbb{G}$.

Let $\mathbf{j}$ be the all-ones vector. $\mathbf{j}$ is the eigenvector of
$\mathbf{A}$ associated with the eigenvalue $d$.

Define $\mathbf{y_s}$ such that
\begin{equation}
	\mathbf{x_s = \frac{\mid\mathbb{S}\mid}{2n}j\;+\;y_s}
\end{equation}

$\mathbf{y_s}$ has $\mid \mathbb{S} \mid$ co-ordinates equal to $1-\frac{\mid \mathbb{S} \mid}{2n}$
and $2n - \mid \mathbb{S} \mid$ co-ordinates equal to $-\frac{\mid \mathbb{S} \mid}{2n}$ and
$\mathbf{y_s}$ is orthogonal to $\mathbf{j}$. Therefore, we can write
\begin{eqnarray*}
	\mathbf{x_s^TAx_s} & = & \frac{{\mid \mathbb{S} \mid}^2}{4n^2}\;d\;\mathbf{j.j}\quad+\quad\mathbf{y_s^TAy_s}
\end{eqnarray*}

Since $\mathbf{j.j}=2n$, we have,

\begin{equation}
	\mathbf{x_s^TAx_s}  =  \frac{{\mid \mathbb{S} \mid}^2}{2n}\;d\quad+\quad\mathbf{y_s^TAy_s}
\end{equation}

Now, $\mathbf{y_s}$ is orthogonal to $\mathbf{j}$ and the eigenspace associated to the eigenvalue $d$
is of dimension 1 ($\mathbb{G}$ is connected). Therefore, we have,
$\mathbf{y_s^TAy_s \leq \lambda {\parallel y_s\parallel}^2}$ where $\lambda$ is the second largest eigenvalue of $\mathbf{A}$.
Considering the structure of $\mathbf{y_s}$ as explained above, we have,
\begin{eqnarray*}
	{\parallel\mathbf{y_s}\parallel}^2 & = & \mid\mathbb{S}\mid(1-\frac{\mid\mathbb{S}\mid}{2n})^2\;+\;(2n-\mid\mathbb{S}\mid)(\frac{\mid\mathbb{S}\mid}{2n})^2 \\
 & = & \mid\mathbb{S}\mid\quad-\quad\frac{{\mid\mathbb{S}\mid}^2}{2n}
\end{eqnarray*}
	Since we are looking for subgraphs in which the degree of each vertex is at least a certain value
(say $\gamma$), after combining the various equations and inequalities above, we get,
\begin{eqnarray*}
	\gamma\mid\mathbb{S}\mid & \leq & \mathbf{x_s^TAx_s} \\
			& = & \frac{{\mid \mathbb{S} \mid}^2}{2n}\;d\quad+\quad\mathbf{y_s^TAy_s} \\
			& \leq & \frac{{\mid \mathbb{S} \mid}^2}{2n}\;d\quad+\quad\lambda {\parallel\mathbf{y_s}\parallel}^2 \\
			& = & \frac{{\mid \mathbb{S} \mid}^2}{2n}\;d\quad+\quad\lambda(\mid\mathbb{S}\mid-\frac{{\mid\mathbb{S}\mid}^2}{2n})
\end{eqnarray*}
Finally, since $\mid\mathbb{S}\mid > 0$, we can cancel it from both sides and the expression that we arrive at is,
\begin{equation}
	\mid\mathbb{S}\mid \geq \frac{2n(\gamma - \lambda)}{d - \lambda}
\end{equation}
Because of duality of points and hyperplanes, we get $\xi=\frac{\mid\mathbb{S}\mid}{2}$. Thus,
\begin{equation}
	\xi \geq \frac{n(\gamma - \lambda)}{d - \lambda}
\end{equation}
The above formula is also stated in \cite{hoholdt} in the context of finding the minimum distance of the code proposed by him using $\mathbb{PG}(2,q)$.
In our case, the above formula should be used only for $\epsilon \geq 15$.
For all practical values of $\epsilon$ ($<$ 15), the combinatorial methods detailed
earlier give a  very tight bound.
\end{appendix}

\vfill


\end{document}